\documentclass[onecolumn,floatfix,11pt,nofootinbib]{revtex4}
\usepackage[dvips]{epsfig}
\usepackage[english]{babel}
\usepackage{bbm}
\usepackage{verbatim}
\usepackage{array}
\usepackage{float}
\usepackage{bm} 
\usepackage{amsmath}
\usepackage{amsfonts}
\usepackage{amssymb}
\usepackage{hyperref}
\usepackage[dvipsnames]{xcolor}
\hypersetup{colorlinks=true,linkbordercolor=Red,linkcolor=Red, citecolor=Red}
\usepackage{indentfirst} 
\usepackage{float}
\usepackage{graphicx} 
\usepackage{lipsum}		    
\usepackage{amssymb, amsmath, amsthm}
\usepackage{mathrsfs}
\usepackage{dsfont}
\newtheorem{definition}{Definition}

\newtheorem{theorem}{Theorem}
\newtheorem{lemma}{Lemma}

\usepackage{bbold}

\begin{document}
	
	\title{Local versus global subtleties of projective representations}
	
	\author{J. M. Hoff da Silva} 
	\email{julio.hoff@unesp.br}
	\affiliation{Departamento de F\'isica, Universidade
		Estadual Paulista, UNESP, Av. Dr. Ariberto Pereira da Cunha, 333, Guaratinguet\'a, SP,
		Brazil.}
	
	\author{J. E. Rodrigues} 
	\email{jose.er.batista@unesp.br}
	\affiliation{Departamento de F\'isica, Universidade
		Estadual Paulista, UNESP, Av. Dr. Ariberto Pereira da Cunha, 333, Guaratinguet\'a, SP,
		Brazil.}

\begin{abstract}
In this short review, we pay attention to some subtleties in the study of projective representations, contrasting local to global properties and their interplay. The analysis is exposed rigorously, showing and demonstrating the main necessary theorems. We discuss the implementation of useful algebraic topology tools to characterize representations.
\end{abstract}		
\maketitle	
	
\section{Introduction}	
In answering the honorific invitation to contribute with a manuscript to the {\it Festschrift} for Professor Ruben Aldrovandi, it is our pleasure to start recalling one of his recommendations dating back to 2004 or so, when asked about the physical meaning of a particle. After a long and detailed explanation about the mathematical structure underlying the representation idea and (quite probably) guessing the first author's lack of necessary background, Professor Aldrovandi recommended a ``judicious reading of Bargmann's paper'' \cite{bar}. This short review presents some subtleties of projective representations whose roots can be traced back to this seminal work. 
	
In another seminal work \cite{wig0}, E. P. Wigner provided a consistent approach to representing the Poincar\`e group in the Hilbert space without even assuming the continuity of the representation but instead demonstrating it. As a result, the very concept of a particle arises firmly supported by the robustness of Wigner's mathematical approach: a particle is an irreducible representation of the Poincar\`e group, connecting once and for all the spacetime symmetries to the particles it supports. In the process, it could not be less asked from the physical point of view: all the results are obtained from the requirement that quantum physics gives the same results here and there, today and tomorrow, provided the same conditions, i.e., it only requires symmetry for quantum processes. Bargmann's work systematizes some of Wigner's steps concerning group representations to provide general results whose particular cases encompass Lorentz and Poincar\`e group representations in the Hilbert space. In particular, these two cases follow more simply (although complex) than in the Wigner approach. 

The advance of algebraic topology and the usage of some of its tools greatly simplify the process of studying whether a representation is genuine or projective. In this regard, \v{C}ech cohomology is particularly useful. When defined upon a continuous manifold, the standard formalism developed by Bargmann is straightforwardly handled with the aid of \v{C}ech cohomology elements. However, Bargmann's approach stops being helpful when this is not the case, however, and the group-associated manifold is not topologically trivial. However, elements of \v{C}ech cohomology are still helpful to study representations.          

This manuscript is organized as follows: starting, in section II, from the basic concepts underlining the mathematical theory of unitary representations of continuous groups, we move to the idea of local exponents and the demonstration of its differentiability in section III. The definition of local groups and the conditions under which local results can be claimed valid globally are presented in section IV. Section V delves into some useful tools from algebraic topology, whose generality and manageability allow interesting connections. In the section VI we conclude.    	        
	
\section{Basic Background}	
	
States $\Psi$ in Quantum Mechanics are equivalence classes, so-called rays, of vectors in a given Hilbert space. Each vector $\psi$ itself is a representative of the class it belongs to, i.e., $\psi\in\Psi$, and each representative of a given ray differs from another representative by a unimodular complex phase. For two representatives $\psi\in\Psi$ and $\phi\in\Phi$, the transition probability of a state $\Psi$ to $\Phi$ is given by $|(\psi,\phi)|^2$, and hence the inner product between rays is naturally given by $\Psi\cdot\Phi=|(\psi,\phi)|$. Two descriptions of a given quantum mechanical system are isomorphic if, and only if, there exists a one-to-one correspondence $\Psi\leftrightarrow \Phi$ between rays preserving transition probabilities: $\Psi_1\cdot\Psi_2=\Phi_1\cdot\Phi_2$. This isomorphism encompasses relevant physical situations as the description of the same quantum phenomena in two distinct inertial frames.  

In analogy to vectorial rays, it is possible to define an operator ray, say $\mathcal{U}$, as the set of all operators $\tau U$, with $|\tau|=1$, for a fixed operator $U$, also called a representative of $\mathcal{U}$. It is due to Wigner \cite{wig1} a theorem stating that an isomorphic ray correspondence defines a unitary and linear (or anti-unitary and anti-linear) operator. For topological continuous groups $G$ (the standard case for this tutorial review), operators representing group elements belonging to its identity component are necessarily unitary. Besides, the product of two operator rays $\mathcal{U}\mathcal{V}$ is defined as the ray composed by all the products $UV$ with $U\in\mathcal{U}$ and $V\in\mathcal{V}$. Let $a\in G$ be a group element, $\mathcal{U}_a$ its isomorphic ray correspondence, and $U_a$ a representative. From the usual group representation rule $\mathcal{U}_a\mathcal{U}_b=\mathcal{U}_{ab}$ one has $U_a U_b=\omega(a,b)U_{ab}$, with $|\omega(a,b)|=1$. As it can be seen, selecting new representatives given by $U'_a=\phi(a)U_a$ and $U'_b=\phi(b)U_b$ ($|\phi(a)|=1=|\phi(b)|$) we have, from  $U'_a U'_b=\omega'(a,b)U'_{ab}$ and using the linearity of the representation, $[\phi(a)\phi(b)\omega(a,b)-\phi(ab)\omega'(a,b)]U_{ab}=0$ motivating the following definition:

\begin{definition}
	Two representation factors, $\omega(a,b)$ defined in a neighborhood $\Sigma$ and $\omega'(a,b)$ defined in $\Sigma'$ are said equivalent if, in $\Sigma_0\subset\Sigma\cap\Sigma'$, they are related by 
	\begin{equation}
	\omega'(a,b)=\frac{\phi(a)\phi(b)}{\phi(a,b)}\omega(a,b).\label{1}
	\end{equation}
\end{definition}	     

For those systems whose group realization can be performed with $\omega=1$, the representation is said to be genuine, otherwise it is called projective. The theory developed by Bargmann \cite{bar} answers quite judiciously whether a given (continuous) group representation is genuine or projective. In order to see the part of this theory relevant for this short review, we shall recall the definition of continuity of a given representation:

\begin{definition}\label{d2}
	A ray representation of a group $G$ is continuous if the following condition is reached: for every element $a\in G$, every ray $\Psi$ in the Hilbert space, and every $\epsilon\in\mathbb{R}^*_+$, there exists a neighborhood $\Sigma$ of $a$ in $G$ such that the distance between $\mathcal{U}_a\Psi$ and $\mathcal{U}_b\Psi$ is less than $\epsilon$ if $b\in\Sigma$.   
\end{definition} This concept of continuity in this context is profound and brings several deep consequences in representation theory \cite{neeb}. By now, it is enough to highlight two aspects related to this definition: firstly \cite{bar} it refers to the theoretical statement that probability transitions vary continuously with the group element. Secondly, as Wigner proved in 1939 \cite{wig0}, it allows for a special selection of representatives strongly continuous (in the sense of Definition \ref{d2}), the so-called {\it admissible representatives}\footnote{The theorem stating this possibility was revisited by Bargmann \cite{bar}. See also \cite{peq} for a broad discussion about this proof.}. Let, then, $\{U_a\}$ be an admissible set of representatives of a continuous group $G$ in a given neighborhood $\Sigma$ of the identity $e\in G$, so that $U_e=1$ and $a,b$, and $ab$ are in $\Sigma$. Therefore $U_{ab}$ is well defined and belongs to the same ray as $U_aU_b$. Hence $U_aU_b=\omega(a,b)U_{ab}$ and, obviously, $\omega(e,e)=1$. Moreover, the associative law $(U_aU_b)U_c=U_a(U_bU_c)$ leads to 
\begin{equation}
\omega(a,b)\omega(ab,c)=\omega(b,c)\omega(a,bc).\label{2}
\end{equation} 

It can be shown that the continuity of admissible representatives implies the continuity of $\omega(a,b)$ factors \cite{bar,peq}. Now we shall state a few important definitions.
\begin{definition}
	  Every complex and continuous function $\omega(a,b)$ ($|\omega(a,b)=1|$) defined for the $a,b$ elements of some neighborhood $\Sigma$ is a {\it local factor} of $G$ defined in $\Sigma$ if $\omega(e,e)=1$ and Eq. (\ref{2}) is valid whenever $ab$ and $bc$ belong to $\Sigma$.  
 \end{definition}

\begin{definition}
	If $\Sigma$ coincides with $G$, such that (\ref{2}) is valid in all group, $\omega$ is then called factor of $G$.
\end{definition} 

When treating a given projective representation, the reader is probably more familiar with exponential terms. In the next section, we shall delve into this more familiar and easily handled notation (and its consequences).  

\section{Local exponents and Differentiability}

As mentioned, it is often beneficial\footnote{See a remark on such identification at the beginning of the theorem's (\ref{t1}) proof.} to write $\omega(a,b)=e^{i\xi(a,b)}$, $\xi\in\mathbb{R}$, so that the previous definitions can be restated as 
\begin{definition}\label{le}
	A local exponent of a group $G$ defined in a neighborhood $\Sigma$ is a real and continuous function $\xi(a,b)$ defined for every $a,b$ of $\Sigma$ satisfying
	\begin{itemize}
		\item $\xi(e,e)=0$, so that $\omega(e,e)=1$;
		\item $\xi(a,b)+\xi(ab,c)=\xi(b,c)+\xi(a,bc)$, $ab,bc\in\Sigma$, so that Eq. (\ref{2}) is valid. 
	\end{itemize} 
 Finally, if $\Sigma$ coincides with $G$, $\xi$ is called an exponent of $G$. 
\end{definition} 

In general, the phases of a projective representation are independent of the physical state upon which it acts \cite{wei}. Nevertheless, the proof for this statement needs it is always possible to prepare a physical state represented as the sum of two linearly independent states. This requirement cannot always be accomplished, and, for those systems, the phases can also depend on dynamical labels. We shall keep our notation without referencing this generality by now, but one should bear in mind this additional generality. Now, from Eq. (\ref{1}), taking $\omega=e^{i\xi}$ and $\phi(a)=e^{ix(a)}$, we have 
\begin{eqnarray}
\xi'(a,b)=\xi(a,b)+\Delta_{a,b}[x],\label{qua}
\end{eqnarray} where 
\begin{equation}
\Delta_{a,b}[x]:= x(a)+x(b)-x(ab).\label{se}
\end{equation} If $x(a=e)=0$ it is fairly simple to see that $\xi'(e,e)=0$. Besides, from Eq. (\ref{qua}),
\begin{equation}
\xi'(a,b)+\xi'(ab,c)=\xi(a,b)+\Delta_{a,b}[x]+\xi(ab,c)+\Delta_{ab,c}[x],
\end{equation} which, with the aid of the second item of Def. (\ref{le}), reads 
\begin{equation}
\xi'(a,b)+\xi'(ab,c)=\xi(b,c)+\xi(a,bc)+\Delta_{a,b}[x]+\Delta_{ab,c}[x].
\end{equation} Finally, noticing that $\xi(k,l)=\xi'(k,l)-\Delta_{k,l}[x]$ and $\Delta_{a,b}[x]+\Delta_{ab,c}[x]-\Delta_{b,c}[x]-\Delta_{a,bc}[x]=0$, we arrive at 
\begin{equation}
\xi'(a,b)+\xi'(ab,c)=\xi'(b,c)+\xi'(a,bc),\nonumber
\end{equation} that is, all the conditions present in Def. (\ref{le}) are filled and the following statement can be written: if $\xi$ is a local exponent defined in $\Sigma$ and $x(a)$ a continuous real function in $\Sigma^2$ (a neighborhood consisting of products $ab$) such that $x(e)=0$, then $\xi'$ defined by (\ref{qua}) is also a local exponent in $\Sigma$. All that motivates the following definition:
\begin{definition}
	Two local exponents $\xi$ and $\xi'$ defined in $\Sigma$ and $\Sigma'$, respectively, are equivalent if Eqs. (\ref{qua}) and (\ref{se}) are valid in some neighborhood $\Sigma_0\subset (\Sigma\cap\Sigma')$, where $x(a)$ is a continuous real function defined in $\Sigma^2$. 
\end{definition} We observe that the equivalence $\xi\equiv \xi'$ produced by Eq. (\ref{qua}) is a formal equivalence relation, i.e., symmetric, reflexive, and transitive.

We shall now pay some attention to an important aspect of this construction, the differentiability of local exponents for continuous groups. Differentiable means a term with continuous partial derivatives of all orders concerning the group elements or coordinates. The proof presented in \cite{bar} is an adaptation of a method for approaching some more or less known (at that time) results about Lie groups shown in Ref. \cite{iwa}. Let us revisit it in detail.  
\begin{theorem}\label{uia}
 In a Lie group every local exponent is equivalent to a differentiable local exponent.  	
\end{theorem}	 
\begin{proof}
	As it is well known \cite{barut}, every Lie group $G$ supports a Haar invariant measure. Let $da$ and $d'a$ denote the left and right invariant measures in $G$, respectively. Now define two real functions, $g$ and $g'$, diffentiable in $\Sigma$, and null everywhere outside $\Sigma'\subset \Sigma$. Besides, the functional form of these functions must respect\footnote{At first sight, this set of requirements may appear too restrictive. That is not the case, however. It is possible to envisage a plethora of functions filling all requirements.} 
	\begin{equation}
	\int_G g(a)da=1=\int_G g'(a)d'a.\label{iw1}
	\end{equation} Note that the integrals are trivial in $G$, being non-null (and equal to $1$) only in $\Sigma'$. 
	
	Define two equivalence relations by 
	\begin{eqnarray}
	\xi'(a,b)=\xi(a,b)+\Delta_{a,b}[x] \nonumber, \\ x(a)=-\int_G\xi(a,k)g(k)dk \label{iw2},
	\end{eqnarray} where $a,b\in \Sigma_1$, and   \begin{eqnarray}
	\xi''(a,b)=\xi'(a,b)+\Delta_{a,b}[x'] \nonumber, \\ x'(a)=-\int_G\xi'(l,a)g'(l)d'l \label{iw3},
	\end{eqnarray} where $a,b\in \Sigma_2\subset \Sigma_1$. Now take $\Sigma'\subset \Sigma_2$ and consider Eqs. (\ref{iw2}) starting with $\xi'(a,b)=\xi(a,b)\cdot 1+\Delta_{a,b}[x]$ writing (in $\Sigma'$) $1=\int_G g(k) dk$. Using Eq. (\ref{se}) we have 
	\begin{equation}
	\xi'(a,b)=\int_G\big\{\xi(a,b)-\xi(a,k)-\xi(b,k)+\xi(ab,k)\big\}g(k)dk,
	\end{equation} and from the very definition (\ref{le}), which reproduces Eq. (\ref{2}), we are left with 
	\begin{equation}
	\xi'(a,b)=\int_G\big\{\xi(a,bk)-\xi(a,k)\big\}g(k)dk.\label{eve}
	\end{equation} In the first integral, the relabel $k\mapsto b^{-1}k$ produces
	\begin{equation}
	\int_G \xi(a,bk)g(k)dk\longrightarrow \int_G \xi(a,k)g(b^{-1}k)dk,\nonumber
	\end{equation} since $dk$ is a Haar measure. Back to Eq. (\ref{eve}) we arrive at 
	\begin{equation}
	\xi'(a,b)=\int_G \xi(a,k)\big\{g(b^{-1}k)-g(k)\}dk.\label{iw4}
	\end{equation} The above steps may be followed for Eqs. (\ref{iw3}) in a similar fashion, resulting in 
	\begin{equation}
	\xi''(a,b)=\int_G \xi'(l,b)\big\{g'(la^{-1})-g'(l)\big\}d'l \label{iw5}.
	\end{equation} Finally, inserting Eq. (\ref{iw4}) into (\ref{iw5}) leads to 
\begin{equation}
 \xi''(a,b)=\int_G \int_G \xi(l,k)\big\{g(b^{-1}k)-g(k)\big\}\big\{g'(la^{-1})-g'(l)\big\}dkd'l
\end{equation} and since $\xi''(a,b)$ depends on $a$ and $b$ only via $g$ and $g'$, differentiability of $\xi''$ is inherited from these two functions. Lastly, since $\xi\simeq \xi''$ the theorem is proved. 
\end{proof}	

A similar procedure establishes differentiability for $x$ functions, as seen in the sequel.

\begin{lemma}
	If two differentiable exponents of a Lie group are equivalents, then $\xi'=\xi+\Delta[x]$ in a suitably chosen neighborhood with $x$ differentiable. 
\end{lemma}	 
\begin{proof}
Since $\xi'$ and $\xi$ are differentiable in a neighborhood, so it is $\Delta[x]$ and $\eta$ defined by $\eta(a)=\int_G\Delta_{a,k}[x]g(k)dk$, with $g$ defined as in the theorem (\ref{uia}). Now define $\bar{x}(a)=x(a)-\eta(a)$, i.e.
\begin{eqnarray}
\bar{x}=x(a)\cdot 1-\int_G\big\{x(a)+x(k)-x(ak)\big\}g(k)dk.\label{emi}
\end{eqnarray} Expressing $1=\int_G g(k)dk$ the expression above reduces to $\bar{x}=\int_G\big\{x(k)-x(ak)\big\}g(k)dk$. Now, taking $k\mapsto a^{-1}k$ in the second integral, we are left with 
\begin{eqnarray}
\bar{x}(a)=\int_G x(k)\big\{g(k)-g(a^{-1}k)\big\}dk,
\end{eqnarray} showing the dependence of $\bar{x}$ on $a$ through $g$ and evincing the differentiability of $\bar{x}$ and, as a consequence, of $x$. 
\end{proof}

We end this section calling attention to the fact that differentiability implies continuity, a fact to be appreciated in Theorem (\ref{tr}). 

\section{Considerations about the local group and extensions}

It is possible to explore further the ray characteristic of a given admissible set of representatives. Operators belonging to $\mathcal{U}_a$ are of the form $e^{i\theta}U_a$, with $\theta\in\mathbb{R}$. Therefore, the representation reads
\begin{equation}
(e^{i\theta}U_a)(e^{i\theta'}U_b)=e^{i(\theta+\theta')}\omega(a,b)U_{ab}=e^{i(\theta+\theta'+\xi(a,b))}U_{ab},
\end{equation} suggesting the following definition:
\begin{definition}
	Let $\xi$ be a local exponent of $G$ defined in $\Sigma$. The {\it local group $H$} is composed by pairs $\{\theta,a\}$, $\theta\in\mathbb{R}$ and $a\in\Sigma^2$, with composition rule given by 
	\begin{equation}
	\{\theta_1,a\}\{\theta_2,b\}=\{\theta_1+\theta_2+\xi(a,b),ab\}.\label{h}
	\end{equation}  
\end{definition} The topological space associated with $H$ is the product $\mathbb{R}\times \Sigma^2$, and its group properties are straightforwardly verified from the definition right above and Def. (\ref{le}). In particular, the identity element $\bar{e}$ of $H$ is $\bar{e}=\{0,e\}$, while the inverse reads $\{\theta, a\}^{-1}=\{-(\theta+\xi(r,r^{-1}),r^{-1})\}$. There is also a one-parameter subgroup, say $C$, belonging to the center of $G$, characterized by elements of the form $\{\theta,e\}$. It can be seen (\cite{bar}, and \cite{peq} for a comprehensive discussion) that $H/C\simeq G$. Take two local exponents, $\xi$ and $\xi'$, defined in $\Sigma$. The mapping 
\begin{eqnarray}
   &&\varphi:\left. H\rightarrow H' \right.\nonumber\\&&
   \left.\{\alpha,a\}\mapsto \varphi(\{\alpha,a\})=\{\alpha-x(a),a'=a\}\right.,\label{oia}
\end{eqnarray} where $x(a)$ is the real function appearing in the equivalence relation (\ref{se}), establish the isomorphism $H\simeq H'$ locally.  

In general, given a set of admissible representatives defined in a neighborhood $\Sigma$ of $G$, extending it for a set defined everywhere in $G$ is impossible. This observation is crucial: it happens more often than never for the phase of the projective representation to be equivalent to zero locally. However, this is different globally. Topological obstructions in the group manifold may forbid local results to be valid everywhere. We shall now revisit some important theorems to give a flavor of this point. 

\begin{theorem}\label{t1}
	Let $G$ be a connected and simply connected group, $\omega$ a group factor, and $\,$ $\mathcal{U}_a$ a continuous ray representation such that the local factor defined by an adequately chosen set of admissible representatives $\{U_a\}$ coincides with $\omega$ em some neighborhood $\Sigma$. Then exists an admissible set of representatives $\{U'_a\}$ uniquely determined for the whole group such that $U'_aU'_b=\omega(a,b)U'_{ab}$ and $U_a=U'_a$ in some neighborhood $\Sigma'\subset\Sigma$. 
\end{theorem} 
\begin{proof}
 Being $G$ connected and simply connected, there is a unique continuous $\xi(a,b)$ function, solution of $e^{i\xi(a,b)}=\omega(a,b)$ $\forall a,b\in G$. This is an exponent of $G$ and implies that $H$ is also a connected and simply connected group. Let $V_{(\theta,a)}=e^{i\theta}U_a$ be the representation of an element $\{\theta,a\}\in H$. The composition law in $H$ is ordinary (not projective), and hence, there exists a strongly continuous (unitary) representation $W(\theta,a)$ of $H$, which coincides with $V_{(\theta,a)}$ in some neighborhood. Here is the reason: following \cite{pt}, every element, say $a$, of a connected group may be written as the product $a=a_1a_2\cdots a_n$ of $n$ (finite) elements of a neighborhood. Therefore, $V_a=\Pi_i V_i$, $i=1,\cdots,n$. Besides, for a simply connected group it is always possible to have $V_{a_1}\cdots V_{a_n}=V_{a'_1}\cdots V_{a'_m}$ if $a_1\cdots a_n=a'_1\cdots a'_m$ \cite{pt}. Therefore, as the product $V_{a_1}\cdots V_{a_n}$ depends on the element $a$ (and not on the partition itself), from defining $W_a=V_{a_1}\cdots V_{a_n}$ for $a\in \Sigma'$ the group composition property follows straightforwardly (the uniqueness also follows from $a_1\cdots a_n=a'_1\cdots a'_m$ and continuity of $W's$ is inherited from $V's$). Returning to our notation, there is, then, a unitary, continuous representation $W_{(\theta,a)}$ of $H$ coinciding with $V_{(\theta,a)}$ in $\Sigma'$. 
 
 Now, let $\Sigma'\subset \mathbb{R}\times\Sigma^2$ be such that $|\theta|<k$ in $\Sigma'$, for a suitable\footnote{It is possible to refine the neighborhood concept for this proof. However, this approach is enough for our general purposes.} $k\in\mathbb{R}$. For every $\theta$ there exists $\beta=\theta/n$, $n\in\mathbb{N}^*$, such that $\{\theta,e\}=\{\beta,e\}^n$, or, in terms of our previous analysis, $W_{(\theta,e)}=(V_{(\beta,e)})^n=e^{i\theta}\,1$. Take $U'_a=W_{(0,e)}$ $\forall a$. By means of Eq. (\ref{h}), $\{\theta,a\}=\{\theta,e\}\{0,a\}$, that is $W_{(\theta,a)}=W_{(\theta,e)}W_{(0,a)}=e^{i\theta}U'_a$ and, as a consequence, $U_a=U'_a$ in $\Sigma'$. Finally, again with the aid of (\ref{h}) $\{0,a\}\{0,b\}=\{\xi(a,b),ab\}$, i.e., $U'_a U'_b=e^{i\xi(a,b)}U'_{ab}=\omega(a,b)U'_{ab}$.
\end{proof} 

Since we are interested in investigating the global validity of local properties, it is only natural to define the next concept.
\begin{definition}
	An exponent $\xi_1$ of $G$ is called an extension of $\xi$ if they are equal in some neighborhood.
\end{definition} The following result asserts that the extensions of local equivalent exponents are also equivalent. 

\begin{theorem}\label{tr}
Let $\xi$ and $\xi'$ be two equivalent local exponents of a connected and simply connected group $G$, such that $\xi'=\xi+\Delta[x]$ in some neighborhood, and assume that the exponents $\xi_1$ and $\xi'_1$ of $G$ are extensions of $\xi$ and $\xi'$, respectively. Then, $\xi'_1(a,b)=\xi_1(a,b)+\Delta_{a,b}[x_1]$ for every $a,b \in G$, where $x_1(a)$ is continuous in $a$ and $x_1=x$ in some neighborhood $\Sigma'$.    
\end{theorem}
\begin{proof}
It is clear that both exponents $\xi'_1$ and $\xi_1$ define two connected and simply connected local groups $H'_1$ and $H_1$, respectively, for which the composition rule is given by (\ref{h}). As mentioned, the mapping $\varphi:H_1\rightarrow H'_1$ given by (\ref{oia}) establishes a local isomorphism. In the Ref. \cite{pt}, there is an important result: the mapping $\varphi$ may be uniquely extended to an isomorphism $\varphi^*$ of $H_1$ and $H'_1$ as a whole (and not only locally) if $H_1$ is locally connected and simply connected\footnote{Intuitively, for a connected and simply connected group, the local isomorphism may be settled everywhere around the identity element, without any distinction. This is certainly not the case if topological obstructions are in order in the group manifold.}, with $\varphi^*=\varphi$ in $\Sigma'\subset \Sigma$. 

Of course  $\varphi^*(\theta,e)=\{\theta,e\}$ $\forall \theta$. Now take $\varphi^*(0,a)=\{-x_1(a),g(a)\}$ where $g(a)$ is to be found. From $\{\theta,a\}=\{\theta,e\}\{0,a\}$, we have $\varphi^*(\theta,a)=\{\theta-x_1(a)+\xi_1'(e,a),g(a)\}$. However, from the general rule $U_a U_b=e^{i\xi(a,b)}U_{ab}$ it is straightforward to see that if, for instance, $a=e$ then $1U_b=e^{i\xi'_1(e,b)}U_{eb=b}$, forcing the conclusion that\footnote{Just as all local exponent: $\xi'_1(e,b)=0=\xi'_1(a,e)=\xi'_1(e,e)$.} $\xi'_1(e,b)=0$. Therefore $\varphi^*(\theta,a)=\{\theta-x_1(a),g(a)\}$. We can now particularize the neighborhood to $\Sigma'$ to fix $g(a)=a$ and, then, extend it to all element $a$. This leads to  
\begin{equation}
\varphi^*(0,a)\varphi^*(0,b)=\{-x_1(a),a\}\{-x_1(b),b\}=\{-x_1(a)-x_1(b)+\xi'_1(a,b),ab\},\nonumber
\end{equation} and, on the other hand, $\varphi^*(0,ab)=\{-x_1(ab),ab\}$, from which\footnote{Along with the fact that $\varphi^*(0,a)\varphi^*(0,b)=e^{i\xi_1(a,b)}\varphi^*(0,ab)$.} we arrive at $\xi'_1(a,b)=\xi_1(a,b)+x_1(a)+x_1(b)-x_1(a,b)$ or, for short, $\xi'_1(a,b)=\xi_1(a,b)+\Delta_{a,b}(x_1)$ on the group as a whole. Finally, being $x_1$ continuous, $x_1=x$ in $\Sigma'$ so as to Eq. (\ref{oia}) dictates.    
\end{proof}

These two theorems are somewhat enough to appreciate the relevance of group topology in the decision between a genuine or projective representation. More often than never, it is possible to show the equivalence between a local exponent to zero in a given neighborhood. Nevertheless, the maintenance of a vanishing local exponent for the whole group requires the group to be connected and simply connected. One must look at the covering group when this is not the case. There is an isomorphism between a given group $G$ and $\tilde{G}/K$, where $\tilde{G}$ is the covering group of $G$ and $K$ is a discrete central invariant subgroup of $\tilde{G}$. If, and only if, $G$ is simply connected, this isomorphism reduces to $G\simeq\tilde{G}$. When $G$ is not simply connected, there is a record, so to speak, of projective representation encoded in $K$. That is indeed the case, for instance, for the Lorentz and Poincar\`e groups in more than two dimensions which are doubly (not simply) connected and representations up to a sign are in order.

As can be seen, the standard approach to the study of projective representations involves group theory, topology, and analysis. While precise and robust, it is certainly demanding. In the next section, we shall glance at some help from algebraic topology.      

\section{A glance at characterizations coming from algebraic topology}
We start introducing the concept of a $n-$cochain following the exposition of Ref. \cite{kiri}:
\begin{definition}
	Let G be a group and M an abelian group, a $(n+1)-$dimensional cochain is a function $\xi(g_0, \cdots, g_n)$ defined on $G \times G \times \cdots \times G$ with values in $M$.
\end{definition}
The set of all $(n+1)-$cochains forms a group $C^{n+1}(G,M)$ and it is possible to define an operator $\delta$ increasing the degree of a given cochain by \cite{ald}
\begin{equation}
\delta: C^{n+1}(G,M) \rightarrow C^{n+2}(G,M)
\end{equation}
\begin{equation}
\delta\xi(g_0, \cdots, g_{n+1}) = \sum_{i=0}^{n+1}(-1)^i\xi(g_0, \cdots, \hat{g_i}, g_ig_{i+1}, \cdots, g_{n+1}),
\end{equation}
where $\hat{g_i}$ stands for suppression of the corresponding (under hat) element. There is still an element of indeterminacy in the expression for $\delta\xi$ concerning the suppression of the last argument factor, but the examples we shall depict will resolve this specific point. The standard nomenclature calls a cochain $\xi'$ by a coboundary of the cochain $\xi$ if $\xi'=\delta\xi$. If $\delta\xi' = 0$, $\xi'$ is called a cocycle. The $\delta$ operator satisfies the fundamental property $\delta \circ \delta = 0$. To verify this property, we shall first see how to operate in simple cases for clarity. For degrees $0$, $1$, and $2$ we have, respectively
\begin{equation}\label{vai}
\begin{split}
&\delta\xi(a,b) = \xi(ab) - \xi(a), \\
&\delta\xi(a,b,c) = \xi(ab,c) - \xi(a,bc) + \xi(a,b), \\
&\delta\xi(a,b,c,d) = \xi(ab,c,d) - \xi(a,bc,d) + \xi(a,b,cd) - \xi(a,b,c).      
\end{split}
\end{equation} Now it is straightforward to see that $\delta(\delta\xi(a,b,c)) =  \delta\xi(ab,c) - \delta\xi(a,bc) + \delta\xi(a,b)$, amounts out to
\begin{equation}\label{jo1}
\delta(\delta\xi(a,b,c)) = \xi(abc) - \xi(ab) - \xi(abc) + \xi(a) + \xi(ab) - \xi(a)
\end{equation} and, therefore, $\delta(\delta\xi(a,b,c)) = 0$. Note that the first term of (\ref{jo1}) has been canceled by the first term arising from the application of $\delta$ in $\xi(a,bc)$, and the other terms follow the very same rule. Bearing in mind this simple prescription, it is possible to see that  
\begin{equation}
\begin{split}
\delta\xi(g_0, \cdots, g_{n+1}) &= \xi(\hat{g_0},g_0g_1,\cdots, g_{n+1}) - \xi(g_0,\hat{g_1}, g_1g_2, \cdots, g_{n+1}) \\ + &\cdots + (-1)^{n+1}\xi(g_0, g_1, \cdots, g_n, \hat{g}_{n+1}),   
\end{split}
\end{equation} leading to 
\begin{equation}
\begin{split}
\delta(\delta\xi(g_0, \cdots, g_{n+1})) &= \delta\xi(\hat{g_0},g_0g_1,\cdots, g_{n+1}) - \delta\xi(g_0,\hat{g_1}, g_1g_2, \cdots, g_{n+1}) \\ + &\cdots + (-1)^{n+1}\delta\xi(g_0, g_1, \cdots, g_n, \hat{g}_{n+1})=0.
\end{split}
\end{equation}

The set of all $n-$dimensional cocycles and the set of all coboundaries of $(n+1)-$dimensional cochains form each one a group, $Z^n(G,M)$ and $B^n(G,M)$, respectively \cite{naka}. This fact motivates the following definition.
\begin{definition} 
	The $n-$th \v{C}ech cohomology group, $\check{H}^n(G,M)$, is defined by
\begin{equation}
\check{H}^n(G,M)=Z^n(G,M)/B^n(G,M).\nonumber
\end{equation}
\end{definition} As it is well known, \v{C}ech cohomology groups are closely related to the De Rham cohomology groups. This relation gains the status of an isomorphism when the base group manifold is a differential manifold \cite{iso}, leading to an interesting parallel. In the De Rham cohomology, dealing with exact and closed forms, Poincar\`e lemma states that every closed form is also exact in a star-shaped domain (i.e., locally). Nontrivial topology engenders global obstructions to this result. Therefore, due to the alluded isomorphism between both cohomologies, this is also the case for \v{C}ech cohomology. This fact and the results previously revised in this manuscript help us understand the usefulness of cohomology groups regarding local versus global properties of phase representations. There is, however, an additional bonus: as we have seen, if the group to be represented is not continuous, the theory presented by Bargmann is of little help. However, cochains may still be defined. We shall return to the main line of this manuscript, making explicit the connection between the elements just described and Bargmann's approach.     

Consider a ray operatorial representation $U_a U_b = e^{i\xi(a,b)}U_{ab}$. Associativity $(U_aU_b)U_c= U_a(U_bU_c)$ implies  
\begin{equation}
\xi(a,b) + \xi(ab,c) = \xi(b,c) + \xi(a,bc),
\end{equation} just as Definition (\ref{le}) stays. From Eqs. (\ref{vai}), however, we already know that $\delta\xi(a,b,c) = \xi(ab,c) - \xi(a,bc) + \xi(a,b)$ and, therefore, one is forced to conclude that $\xi(b,c) = \delta\xi(a,b,c)$, or generically  
\begin{equation}
\xi(a,b) = \delta\xi(c,a,b).
\end{equation} Thus, the associativity of operators in a possibly projective representation induces a phase factor, which is nothing but a coboundary of a $1-$cochain. It shows that we can remove the phase and arrive at a genuine representation if, and only if, $\xi(c,a,b)$ is a coboundary. If this is not the case, i.e., $\xi(a,b)$ is a coboundary which is not a cocycle, then $\check{H}^2(G,M)$ is nontrivial. There is, then, a topological obstruction to the phase elimination, resulting in a projective representation.

We shall finalize pointing out that the survival of a phase in the representation scheme may be accompanied by physical significance. Schwinger showed \cite{sch} that the (Weyl) operators realizing Heisenberg algebra do form a complete basis for all unitary operators, describing symmetric operations in the sense of Wigner \cite{wig0}, evincing quantum kinematics, and giving a complete set of physical degrees of freedom in Quantum Mechanics. The Schwinger operators span a $Z_N\times Z_N$ (where $Z_N$ is the cyclic group) discrete representation of Weyl's realization. In Ref. \cite{ald}, it was shown that in the realm of Schwinger's representation, a fundamental cocycle exists responsible for 1) making the representation globally projective and 2) assigning a pre-symplectic structure in the space state. Within this scope, classical symplectic structure would result from a continuum limit of quantum pre-symplectic structure.  

\section{concluding remarks}

It is possible to assert that the study of group representations in the Hilbert space performs the core of the physical labeling of particles degrees of freedom and furnishes the foundations of investigating quantum kinematics and dynamics. More than a solid mathematical setup to decide whether a representation is genuine, it allows the appreciation of some profound physical results. Apart from the classical and quantum interplay mentioned in the last section, we want to remark on the possibility of frame gauge freedom in Quantum Field Theory as a natural consequence of a (generalized) projective representation Ref. \cite{inp}. The analysis performed in Ref. \cite{inp} may also be applied to nonrelativistic theories where projective representation depends on time. We shall delve into this generalization to the pre-symplectic structure in Ref. \cite{ald} in the near future.     

\subsection*{Acknowledgements}
JMHS thanks to National Council for Scientific and Technological Development -- CNPq (grant No. 307641/2022-8).

\end{document}